\RequirePackage{amsthm}

\RequirePackage{silence}
\documentclass[sn-mathphys-num]{sn-jnl}

\usepackage{graphicx}%
\usepackage{multirow}%
\usepackage{amsmath,amssymb,amsfonts}%
\usepackage{amsthm}%
\usepackage{mathrsfs}%
\usepackage[title]{appendix}%
\usepackage{xcolor}%
\usepackage{textcomp}%
\usepackage{manyfoot}%
\usepackage{booktabs}%
\usepackage{algorithm}%
\usepackage{algorithmicx}%
\usepackage{algpseudocode}%
\usepackage{listings}%

\usepackage{import}
\usepackage{cleveref}
\usepackage{footnotehyper}
\makesavenoteenv{table}
\makesavenoteenv{tabular}
\usepackage[all]{hypcap}
\begin{document}
\vfuzz=10pt

\title{Square packing with $O(x^{0.6})$ wasted area}

\author*[1]{\fnm{Hong Duc} \sur{Bui}}\email{buihd@u.nus.edu}
\theoremstyle{thmstyleone}%
\newtheorem{theorem}{Theorem}
\newtheorem{conjecture}{Conjecture}
\newtheorem{proposition}{Proposition}%
\newtheorem{lemma}{Lemma}%
\newtheorem{corollary}{Corollary}%

\theoremstyle{thmstyletwo}%
\newtheorem{example}{Example}%
\newtheorem{remark}{Remark}%

\theoremstyle{thmstylethree}%
\newtheorem{definition}{Definition}%

\newtheorem{question}{Question}

\abstract{
We show a new construction for square packing, and prove that it is more efficient than previous results.
}

\maketitle

\section{Introduction}

Square packing is a well-studied problem.
Formally, we consider a large square $S$ with side length $x$ and ask what is the maximum number of unit squares
that can be packed without overlap into $S$.

We define $W(x)$ to be the area of wasted space when a square of side length $x$ is packed with unit squares.
In other words, $W(x)$ is $x^2$ minus the maximum number of unit squares that can fit in.
(For convenience, our definition of $W(x)$ is the same as in \cite{Erdos_1975, Chung_2009, wang2016newresultpackingunit}.)

In this article, we are concerned with the behavior of $W(x)$ as $x \to \infty$.

The trivial packing method allows $\lfloor x \rfloor^2$ squares to fit in,
which shows $W(x) \leq x^2-\lfloor x \rfloor^2 \in O(x)$.
However, a long line of research
\cite{Erdos_1975, Chung_2009, wang2016newresultpackingunit}
has given much better bounds, the best one gives $W(x) \in O(x^{0.625})$.
Also, \cite{Chung_2019} claims $W(x) \in O(x^{0.6})$, but \cite{Arslanov2025} shows that
there is a technical error in the calculation of $W_5$, with $W_5$ defined on \cite[p.~8]{Chung_2019}.

This article extends the insights in \cite{wang2016newresultpackingunit} to improve the result,
namely $W(x) \in O(x^{0.6})$.
%

On the opposite direction, \cite{roth1978inefficiency} proved $W(x) \notin o(x^{1/2})$.
This is a lower bound on $W(x)$.
Note that $W(x) \notin \Omega(x^{1/2})$, as $W(x) = 0$ for all positive integer $x$.

In a high level, the packing is done as follows.
First, packing a square is reduced to packing a right trapezoid, this reduction is
described in \Cref{trapezoid_reduction}.
Then, the trapezoid is packed by a certain kind of quadrilateral, this reduction is
described in \Cref{sec:right_trapezoid_packing}.
Finally, \Cref{sec:tightly_packed_quad} describes how to pack such a quadrilateral.

Within these three steps, \Cref{sec:tightly_packed_quad} is the most technical,
\Cref{sec:right_trapezoid_packing} is less technical, and 
\Cref{trapezoid_reduction} is standard, and the idea already appear in prior works.

\subsection{Note on independent discovery}

After putting this manuscript on arXiv, we learnt that an equivalent result was independently
discovered by McClenagan \cite{McClenagan2024,mcclenagan2026}, with a somewhat different method.

The two methods can be compared as follows.

The ``second packing algorithm'' in \cite[Section~3]{mcclenagan2026} is
similar to the packing method in \Cref{sec:tightly_packed_quad},
except that it is specialized to have $\Delta_1 + \Delta_2 = 1$,
so $i_j = j$ for every $j$. In particular, $i_m = m$,
that is, the packing has the same number of rows as columns.

By only analyzing the case where $i_m = m$, the analysis of wasted area becomes significantly simpler.
For comparison, in our article, in the application of \Cref{sec:tightly_packed_quad}
in \Cref{sec:right_trapezoid_packing}, we only use the case where $i_m \in \Theta(m)$.
On the other hand, by analyzing the general case where no restriction between $i_m$ and $m$ are imposed,
we allow for more flexibility on the method in \Cref{sec:right_trapezoid_packing}.
In particular, we don't need to use and analyze the dual packing method.
Besides, by analyzing the general case,
we were able to analyze in \Cref{subsec:limitation}
why having $i_m \in \Theta(m)$ appear to be the best choice with this packing method.

The ``first packing algorithm'' in \cite[Section~2]{mcclenagan2026}
is a vertical reflection of the so-called dual packing method we briefly mentioned in \Cref{subsec:dual_packing}.

The method used to pack a right trapezoid in \Cref{sec:right_trapezoid_packing} of our article
and \cite[Section~4]{mcclenagan2026} are completely different, as explained above.

The reduction from packing a square and packing a right trapezoid
in \Cref{trapezoid_reduction} of our article and \cite[Section~4]{mcclenagan2026}
are the same, and the same as prior works.

\section{Summary of Existing Techniques}
\label{sec:summary_existing}

The fundamental method used to overcome the trivial bound $x^2-\lfloor x \rfloor^2$ is the following.

Consider two parallel lines at a distance $x$ apart,
we wish to pack the space between them with unit squares.
If the trivial packing method is used
as in \Cref{fig:trivial_packing_image},
we would get $\Omega(x-\lfloor x \rfloor)$ area of wasted space per unit distance.
However, if we take vertical stacks of $\lceil x \rceil$ unit squares and tilt them as little as possible
such that they fit between the two lines, the angle of rotation needed will be $O(x^{-1/2})$,
therefore the wasted area is only $O(x^{-1/2})$ per unit distance. See \Cref{fig:fundamental_long_strip}
for an illustration.

\begin{figure}
    \centering
    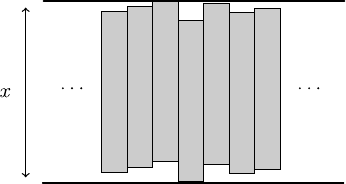
    \caption{Trivial packing method of the space between two parallel lines.}
    \label{fig:trivial_packing_image}
\end{figure}

\begin{figure}
    \centering
    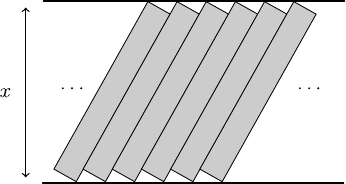
    \caption{Improved packing method of the space between two parallel lines, which reduces the wasted area to $O(x^{-1/2})$ per unit distance.}
    \label{fig:fundamental_long_strip}
\end{figure}

More generally, the height of each vertical stack can be any integer $m\geq x$ such that $m-x \in O(1)$.
On a first-order approximation, the angle of rotation is usually $\sqrt{\frac{2(m-x)}{x}}$,
as can be computed by an application of Pythagorean theorem.

This result can be interpreted as follows:
it provides a family of parallelograms that can be packed with small waste.
These parallelograms can then be assembled together to pack a larger shape.

The reason why this is so commonly used \cite{Erdos_1975, Chung_2009, wang2016newresultpackingunit} is that the parallelogram can have arbitrary (possibly nonintegral) height $x$,
which allows us to eliminate one \emph{obstruction} in square packing---two parallel or almost-parallel sides with distance $x$ apart where $x$ is not an integer.
Unfortunately, it creates another obstruction: the other side of the parallelogram is not perfectly vertical or horizontal.

Previously, there was no good way around this:
for example, \cite[Figure 9]{Chung_2009} cuts away small triangles (denoted $e_i$ in that article) from the side of the slanted trapezoid in order to make the two sides parallel.

In \Cref{sec:tightly_packed_quad},
we will describe another family of quadrilaterals that can be packed similarly tightly (wasted area $=$ perimeter $\times$ slope angle), but have two opposite sides non-parallel.
This is the main primitive that we assemble together
for the final packing.

\section{A Primitive Tightly-packed Quadrilateral}
\label{sec:tightly_packed_quad}

In this section, we prove the existence of a family of quadrilaterals
with the following properties:
\begin{itemize}
    \item Each interior angle is $90^\circ \pm o(1)$.
        Intuitively, the quadrilateral looks almost like a rectangle (unless the side length is too large,
        see \Cref{subsec:triangle_like_trapezoid}).

    \item
        Consider one such quadrilateral. It is approximately an axis-aligned rectangle, as above.
        Let $m$ and $i_m$ be some positive integers such that
        the almost-horizontal sides have length $\approx m$,
        the almost-vertical sides have length $\approx i_m$.
        Let $\theta>0$ such that the top left interior angle has measure $90^\circ - \theta$.

        Then there is a packing of the quadrilateral with total wasted area $\in O((m+i_m) \theta)$.

    \item This packing uses only almost-axis-aligned unit squares, that is, each unit square can be rotated
        by an angle of $o(1)$ to become perfectly axis-aligned.
\end{itemize}

The whole packing is illustrated in \Cref{fig:end_of_packing},
where the quadrilateral to be packed is $ABCD$.

We would like to note that this packing method is inspired from \cite{wang2016newresultpackingunit}.
The connection is described in \Cref{subsec:tightly_packed_quad_inspiration}.

Furthermore, when $i_m \in \Theta(m)$, our packing method
has the asymptotically smallest wasted area
among all packings for this family of quadrilaterals.
This is shown in \Cref{subsec:tightly_packed_lower_bound}.

For the convenience of the reader, all notations used in this section are listed in \Cref{summary_notations}.
Some of the notations (in particular, $\theta$, $\sigma_1$ and $\sigma_2$) are used again in later chapters.

\begin{table}
\begin{tabular}{ccc}
    Symbol                            & Meaning & Note \\ \hline
    $m$ & Number of columns & None \\
    $i_m$ & Number of rows & See $i_j$ below; $i_m \approx \sqrt{\frac{\sigma_2}{\sigma_1}} m$ \\
    $\theta$ & Angle between $AB$ and horizontal & None \\
    $\sigma_1$ & Angle between $Ax$ and $By$ & None \\
    $S_{i, j}$ & Unit square on $i$-th row, $j$-th column & Sloped by $\theta$ \\
    $T_{i, j}$ & \hyperref[def:T_ij]{Modified unit square} & Perfectly axis-aligned \\
    $\Delta_1$ & \hyperref[fig:illustration_delta]{Amount $S_{i, j+1}$ is to the right of $S_{i, j}$} &
    $\Delta_1 = \cos \theta \approx 1-\frac{\theta^2}{2}$ \\
    $\Delta_2$ & \hyperref[fig:illustration_delta]{Amount $S_{i+1, j}$ is to the right of $S_{i, j}$} & $\Delta_2 = \sec(\theta+\sigma_1) \sin \sigma_1 \approx \sigma_1$ \\
    $\Delta_3$ & \hyperref[fig:illustration_delta]{Amount $S_{i+1, j}$ is below of $S_{i, j}$} & $\Delta_3 
    \approx 1 + \frac{\theta(\theta+2 \sigma_1)}{2}$ \\
$\Gamma_{j}$ & \hyperref[def:gamma_j]{Amount $T_{i_j,j-1}$ is below $T_{i_j,j-2}$} & $\Gamma_{j} \approx \frac{\theta^4}{4 \sigma_1}+\theta$ \\
    $i_j$ & Some row index (see \Cref{def:i_j}) & $i_j = \left\lceil \frac{(j-1)(1-\Delta_1)}{\Delta_2} \right\rceil+1 \approx \frac{\theta^2}{2\sigma_1} j$ \\
    $\sigma_2$ & \hyperref[def:sigma_2]{Angle between $AB$ and $CD$} & $4 \sigma_1 \sigma_2 \approx \theta^4$ \\
\end{tabular}
\caption{Summary of notations in \Cref{sec:tightly_packed_quad}.}
\label{summary_notations}
\end{table}

\subsection{Description of the Packing Method}

\subsubsection{Initial Configuration}


\begin{figure}
    \centering
    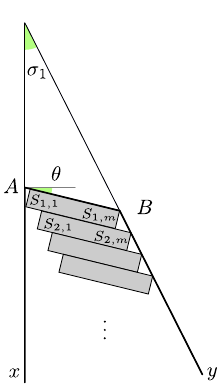
    \caption{First step in the packing method.}
    \label{fig:overall}
\end{figure}

Consider the configuration illustrated in \Cref{fig:overall}.
In words: There are two points $A$ and $B$, the line segment $AB$ is slightly sloped downwards and makes an angle $\theta$ with the horizontal line,
the ray $Ax$ points downward vertically, the ray $By$ points downward and make an angle $\sigma_1$
with the vertical line, such that the rays opposite $Ax$ and $By$ intersect above $A$.

Define a coordinate system such that $x$-axis points to the right, $y$-axis points up,
and point $A$ has $x$-coordinate $0$. (Its $y$-coordinate is unimportant.)

We will pack the area inside the region $xABy$ with horizontal stacks of squares, each $m$ unit squares,
where $m$ is some integer.
We assume we can put a stack of $m$ squares $(S_{1, 1}, \dots, S_{1, m})$, which we denotes $S_{1, \bullet}$,
top right corner touching $B$, top edge lying on the edge $AB$, and bottom left corner touching the ray $Ax$.

Note that the last requirement forces segment $AB$ to have length $m+\tan \theta$.

Then, we keep adding stacks $S_{2, \bullet}, S_{3, \bullet}, \dots$, each having $m$ squares,
top edge parallel to and touching the bottom edge of the previous one,
and top right corner touching the ray $By$.
Number the individual unit squares as in \Cref{fig:overall}.

Define $\Delta_1 = \cos \theta$, $\Delta_2 = \sec(\theta+\sigma_1) \sin \sigma_1$,
$\Delta_3 = \sec(\theta+\sigma_1) \cos \sigma_1$.

\begin{lemma}
    \label{lemma_def_deltas}
    For every $(i, j)$, square $S_{i, j+1}$ is $\Delta_1$ to the right of square $S_{i, j}$,
and square $S_{i+1,j}$ is $\Delta_2$ to the right and
$\Delta_3$ below $S_{i,j}$.
\end{lemma}

See \Cref{fig:illustration_delta} for illustration. This can be shown
by applying the definition of trigonometric functions on the two right triangles depicted in \Cref{fig:2}.

\begin{figure}
    \centering
    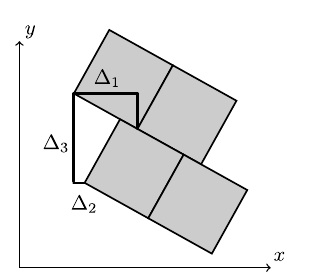
    \caption{Illustration of $\Delta_1$, $\Delta_2$ and $\Delta_3$.}
    \label{fig:illustration_delta}
\end{figure}

\begin{figure}
    \centering
    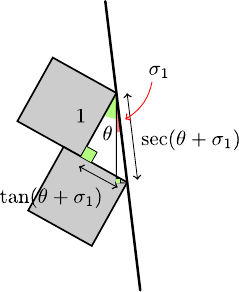
    \caption{Calculation of $\Delta_1$, $\Delta_2$ and $\Delta_3$ from $\theta$ and $\sigma_1$.}
    \label{fig:2}
\end{figure}

Since $\Delta_1<1$, the leftmost point of $S_{1, 2}$ has $x$-coordinate $<1$.
Since $\Delta_2>0$, for sufficiently large $i$, the leftmost point of $S_{i, 2}$ has $x$-coordinate $>1$.

\phantomsection\label{def:i_j}%
For each $j$, define $i_j$ to be the smallest value such that $S_{i_j, j}$ has $x$-coordinate of leftmost point $\geq(j-1)$.
(Clearly $i_1 = 1$. By the argument above, $i_2$ exists and is $>1$.)
Note that $S_{i_j, j}$ has $x$-coordinate of leftmost point $(i_j-1)\Delta_2 + (j-1)\Delta_1$,
therefore $i_j = \left\lceil \frac{(j-1)(1-\Delta_1)}{\Delta_2} \right\rceil+1$.

Note that if $\frac{1-\Delta_1}{\Delta_2} < 1$, some $i_j$ values might coincide.

We have $\Delta_1 \approx 1-\frac{\theta^2}{2}$, $\Delta_2 \approx \sigma_1$,
so $i_j \approx  \frac{\theta^2}{2\sigma_1} j $.

\subsubsection{Modification of the Packing}

We perform some modifications as illustrated in \Cref{fig:introduce_T_squares}. Formal description follows.
For each $j \geq 2$, we remove $S_{i', j-1}$ for all $i' \geq i_j$.
Then we add a perfectly vertical stack of squares with the leftmost point having $x$-coordinate $j-2$,
and the top right point touching the bottom side of $(i_j-1)$-th row.
\phantomsection\label{def:T_ij}%
For each positive integer $\Delta i$, let $T_{i_j+\Delta i-1,j-1}$ be the $\Delta i$-th unit square from the top of this stack.
(So, for example, $T_{i_j,j-1}$ should almost overlap $S_{i_j,j-1}$ that has just been removed.)

\begin{figure}
    \centering
    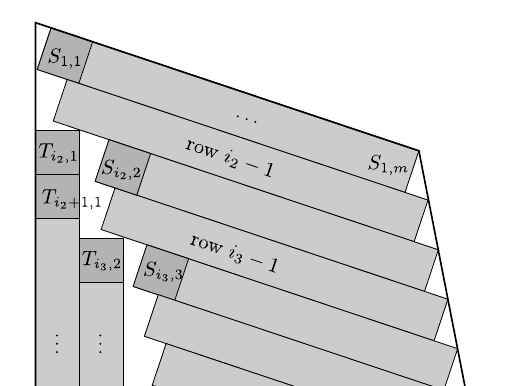
    \caption{Illustration for introduction of the unit squares labeled $T_{i, j}$.}
    \label{fig:introduce_T_squares}
\end{figure}

We note that none of the unit squares $S_{i, j}$ or $T_{i, j}$ overlap.

\phantomsection\label{def:gamma_j}
Define $\Gamma_{j} = (i_j-i_{j-1})(\sec \theta-1) + \tan \theta$.
Since $i_j-i_{j-1}\approx \frac{\theta^2}{2 \sigma_1}$, $\Gamma_{j}\approx \frac{\theta^4}{4 \sigma_1}+\theta$.

\begin{lemma}\label{prop:property_gamma_j}
    For $3 \leq j \leq m$, then $T_{i_j, j-1}$ is $\Gamma_{j}$ below of $T_{i_j, j-2}$.
\end{lemma}
\begin{proof}
\begin{figure}
    \centering
    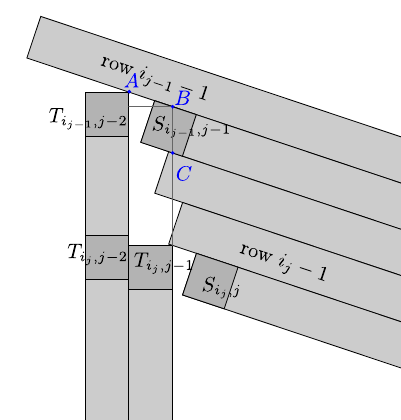
    \caption{Illustration for proof of \Cref{prop:property_gamma_j}. Some important squares are highlighted.}
    \label{fig:4}
\end{figure}
    We construct a few points as in \Cref{fig:4}.
    Formally,
    let $A$ be the corner of $T_{i_{j-1}, j-2}$ that touches the $i_{j-1}-1$-th row.
    Extend the right side of $T_{\bullet, j-1}$ upward to intersect the top and bottom side of row
    $S_{i_{j-1},\bullet}$ at $B$ and $C$ respectively.

    Then, $B$ is $1$ to the right and $\tan \theta$ below $A$.
    Also, $C$ is $\sec \theta$ below $B$.

    Therefore, $T_{i_j, j-1}$ is $(i_j-i_{j-1})\sec \theta + \tan \theta$ below
    $T_{i_{j-1}, j-2}$.
    We also have $T_{i_j, j-2}$ is $(i_j-i_{j-1})$ below $T_{i_{j-1}, j-2}$,
    so we get the desired result.
\end{proof}

We look at the $i_m$-th row of the packing after the modification above.
There are unit squares $T_{i_m, 1}, T_{i_m, 2}, \dots, T_{i_m, m-1}$
being perfectly axis-aligned and $S_{i_m, m}$ sloped by angle $\theta$.
See \Cref{fig:end_of_packing}.

\begin{figure}
    \centering
    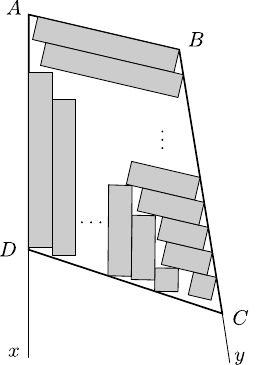
    \caption{Illustration of what happens as the vertical stacks $T_{\bullet, j}$ reaches the right edge.}
    \label{fig:end_of_packing}
\end{figure}

We delete all unit squares with row number greater than $i_m$.
Then we will construct point $C$ on ray $By$ and $D$ on ray $Ax$ such that
all unit squares constructed so far lies inside the quadrilateral $ABCD$.

\phantomsection\label{def:sigma_2}
Define $\sigma_2 = \arctan\big(\frac{1-\Delta_1}{\Delta_2}\cdot(\sec \theta-1)+\tan \theta\big)-\theta$.
Then we get
\[ \sigma_2 \approx \tan(\sigma_2+\theta) - \tan \theta = \frac{1-\Delta_1}{\Delta_2}\cdot(\sec \theta-1)
    \approx \frac{\theta^4}{4 \sigma_1}.
\]

Point $C$ and $D$ are constructed as follows.
The slope of segment $CD$ is determined as follows:
walking from $D$ to $C$, for each unit to the right,
it moves $\tan(\theta+\sigma_2)$ units down.
The vertical position of segment $CD$ will be determined
by the following discussion.

For simplicity of analysis, we delete $S_{i_m, m}$ from the packing.
This results in additional wasted area of $1$ unit.

From the construction above, it follows that:
\begin{proposition}
    Ray $BA$ and $CD$ intersect to the left of vertical line $AD$, and form with each other an angle $\sigma_2$.
\end{proposition}

\begin{figure}
    \centering
    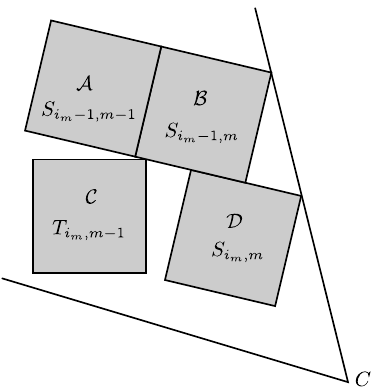
    \caption{Zooming in around vertex $C$.}
    \label{fig:zoom_in_D}
\end{figure}

Consider the unit squares near vertex $C$.
Let
$\mathcal B$, $\mathcal C$ and $\mathcal D$ be the unit squares
$S_{i_m-1, m}$,
$T_{i_m, m-1}$,
$S_{i_m, m}$ respectively, as in \Cref{fig:zoom_in_D}.

Recall that $T_{i_j, j-1}$ is $\Gamma_{j}$ below $T_{i_j, j-2}$.
Therefore $T_{i_m, j-1}$ is also $\Gamma_{j}$ below $T_{i_m, j-2}$,
so for each integer $1 \leq j \leq m-1$, $T_{i_m, j}$ is
\begin{align*}
\Gamma_{3} + \Gamma_{4} + \dots + \Gamma_{j} + \Gamma_{j+1} 
&= (i_{j+1}-i_2)(\sec \theta-1)+(j-1)\tan \theta \\
&= \left( \left\lceil \frac{1-\Delta_1}{\Delta_2} j \right\rceil-\left\lceil \frac{1-\Delta_1}{\Delta_2} \right\rceil \right) (\sec \theta-1)+(j-1)\tan \theta \\
&\leq \left( \frac{1-\Delta_1}{\Delta_2} (j-1) +1 \right) (\sec \theta-1)+(j-1)\tan \theta \\
&= \left( \frac{1-\Delta_1}{\Delta_2}\cdot(\sec \theta-1) +\tan \theta \right)\cdot(j-1) + (\sec \theta-1) \\
&= \tan(\theta+\sigma_2)\cdot(j-1) + (\sec \theta-1) \\
\end{align*}
below
$T_{i_m, 1}$.
As such, if we select the vertical position of segment $CD$ such that
point $D$ is $(\sec \theta-1)$ below the bottom left corner of $T_{i_m,1}$,
then all the $T_{i_m,j}$ unit squares will be above segment $CD$.

\subsection{Analysis of the Wasted Area}

\begin{figure}
    \centering
    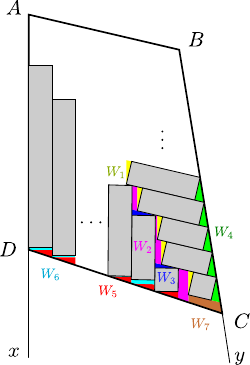
    \caption{Illustrations for analysis of wasted area.}
    \label{fig:waste_analysis}
\end{figure}

There are 7 groups of wasted areas.
See \Cref{fig:waste_analysis} for an illustration.
\begin{itemize}
    \item $W_1$, triangles to the left of each row $S_{i, \bullet}$ (colored \textcolor{yellow!50!black}{yellow}).
        There are $i_m$ of them, each has area $\frac{1}{2} \tan \theta$,
        so the total area is $O(\theta i_m)$.
    \item $W_3$, triangles above each column $T_{\bullet, j}$ (colored \textcolor{blue}{blue}).
        There are $m$ of them, each has area $\frac{1}{2}\tan \theta$,
        so the total area is $O(\theta m)$.
    \item $W_2$, small vertical strips to the left of $W_1$ (colored \textcolor{magenta}{magenta}).
        Since $i_j-i_{j-1}\in O(\frac{\theta^2}{\sigma_1}+1)$ for every $j$,
        the width of each strip is at most $O(\Delta_2 \cdot(\frac{\theta^2}{\sigma_1}+1)) = O(\theta^2+\sigma_1)$,
        so the total area is $O((\theta^2+\sigma_1) \cdot i_m)$.

        Here is more details why the width of each strip
        is at most $O(\Delta_2 \cdot(\frac{\theta^2}{\sigma_1}+1))$.
        For illustration, see \Cref{fig:4}.
        Any horizontal section of the magenta strip is to the left of a row $S_{i, \bullet}$ for some integer $i$.
        Find integer $j$ such that $i_{j-1} \leq i < i_j$.
        Then the width of this horizontal section of the magenta strip
        is the distance between the bottom left corner of $S_{i, j-1}$ and the
        right side of column $T_{\bullet, j-2}$.
        By definition of $i_{j-1}$, the square $S_{i_{j-1}-1, j-1}$ have $x$-coordinate
        of bottom left corner $< (j-2)$.
        Now look at the square right below it, $S_{i_{j-1}, j-1}$.
        By \Cref{lemma_def_deltas},
        the distance between the bottom left corner 
        and the right side of $T_{\bullet, j-2}$ is $< \Delta_2$.
        Similarly, when we look at the square $S_{i, j-1}$, which is $\leq i_j - i_{j-1}$ squares below $S_{i_{j-1}, j-1}$,
        the distance between the bottom left corner of $S_{i, j-1}$
        and the right side of $T_{\bullet, j-2}$ is $< \Delta_2 \cdot (i_j - i_{j-1} + 1)$.

    \item $W_6$, small horizontal strips below each column $T_{\bullet, j}$ (colored \textcolor{cyan}{cyan}).
        There are $m$ of them, the height of each is bounded by $O(\sec \theta-1)$
        (we have shown above during the placement of segment $CD$ that
        $T_{i_m,j}$ is $\leq\tan(\theta+\sigma_2) \cdot(j-1)+(\sec \theta-1)$ below $T_{i_m,1}$,
        using a similar argument we can also show
        $T_{i_m,j}$ is $\geq\tan(\theta+\sigma_2) \cdot(j-1)-(\sec \theta-1)$ below $T_{i_m,1}$),
        so the total area is $O(\theta^2 m)$.
    \item $W_4$, triangles to the right of each row of $S_{i, \bullet}$ (colored \textcolor{green!50!black}{green}).
        There are $i_m$ of them, each has area $\frac{1}{2}\tan(\theta+\sigma_1)$,
        so the total area is $O((\theta+\sigma_1) i_m)$.
    \item $W_5$, triangles below each column of $T_{\bullet, j}$ (colored \textcolor{red}{red}).
        There are $m$ of them, each has area $\frac{1}{2}\tan(\theta+\sigma_2)$,
        so the total area is $O((\theta+\sigma_2) m)$.
    \item $W_7$, unaccounted-for area below $S_{i_m, m}$ (colored \textcolor{brown}{brown}).
        This is $O(1)$ as long as all of $\theta$, $\sigma_1$, $\sigma_2$ are $\in o(1)$.
\end{itemize}

Summing them up, we get the total wasted area to be $
O((\frac{\theta^3}{\sigma_1}+\theta) \cdot m+1)
$.
This is because each angle $\theta$, $\sigma_1$, $\sigma_2$ is $\in O(1)$, and $4 \sigma_1 \sigma_2 \approx \theta^4$. Also this can be written more symmetrically as $O((m+i_m)\cdot \theta+1)$.

\begin{remark}
    \label{remark:alternative-wasted-area-calculation}
    There is an alternative way to calculate the wasted area:
    count the number of squares,
    then subtract that from the area of the trapezoid $ABCD$.

    Clearly the number of squares (including $S_{i_m, m}$) is $m \cdot i_m$.
    If we can show that the area of the trapezoid $ABCD$
    is $\leq m \cdot i_m + O(\theta \cdot (m + i_m) + 1)$,
    we would be able to conclude that the total wasted area
    is $O(\theta \cdot (m + i_m) + 1)$.

    Calculating the area of the trapezoid $ABCD$
    appears to be difficult. See \Cref{alternative_wasted_area_calculation}.
\end{remark}

\subsection{Lower Bound on the Wasted Area}
\label{subsec:tightly_packed_lower_bound}

As mentioned in \Cref{sec:summary_existing}, the existing method for packing the area between two parallel lines
a distance $x$ apart can be interpreted as giving an efficient packing method for a family of parallelograms.
We illustrate such a parallelogram in \Cref{fig:basic_parallelogram}.

If we have to pack the interior of such a parallelogram,
this packing method is in fact asymptotically optimal in certain cases.

\begin{proposition}
    Consider a parallelogram that can be perfectly packed by the method illustrated in 
    \Cref{fig:basic_parallelogram}.
    Let $x$ be its height.
    Suppose $\lceil x \rceil-x \in \Theta(1)$ and the width of the parallelogram is $\in \Theta(x)$.
    Then the internal angle of the parallelogram is $90^\circ \pm \Theta(1/\sqrt x)$,
    and the packing method has wasted area $\Theta(\sqrt x)$.
\end{proposition}

\begin{proof}
    Suppose the horizontal edge is parallel to the $x$-axis.
Let $\theta$ be the tilt of the almost-vertical edge,
and $w \in \Theta(x)$ be the number of almost-vertical stack of squares illustrated in \Cref{fig:basic_parallelogram}.
Then, the wasted area consists of $2w$ small triangles, each have area $\frac{1}{2} \tan \theta$.
Therefore, the total wasted area is $w \tan \theta \in \Theta(x \theta)$.
An easy calculation shows that if $\lceil x \rceil-x \in \Theta(1)$ then $\theta \in \Theta(1/\sqrt x)$.
\end{proof}

\begin{proposition}
    \label{prop_parallelogram_min_wasted_area}
    A parallelogram with height $x$, width $\Theta(x)$,
    and internal angle $90^\circ \pm \theta$ for $\theta \in \Theta(1/\sqrt x)$ has the minimum wasted area
    of a packing with unit squares $\Omega(\sqrt x)$.
\end{proposition}

The following is a more formal version
of the statement of \Cref{prop_parallelogram_min_wasted_area}.
For all functions $w_1(x), w_2(x) \in \Theta(x)$,
$\theta_1(x), \theta_2(x) \in \Theta(1/\sqrt x)$,
there exists a function $a(x) \in \Omega(\sqrt x)$
such that for any parallelogram with height $x$,
width between $w_1(x)$ and $w_2(x)$,
and one internal angle $90^\circ + \theta$ for $\theta_1(x) \leq \theta \leq \theta_2(x)$
the minimum wasted area of a packing of that parallelogram is $\geq a(x)$.

\begin{proof}
We prove that any packing requires wasted area $\in \Omega(x \theta)$.
Pick $\Theta(x)$ equidistant points each
along the left and bottom edge of the parallelogram,
and connect the corresponding points, getting $\Theta(x)$ parallel line segments.
This is depicted with red downwards-sloping line segments in \Cref{fig:basic_parallelogram}.
Let these paths be $\gamma_1$, $\gamma_2$, $\dots$, $\gamma_k$ with $k \in \Theta(x)$.

Because the width and the height of the parallelogram are both $\Theta(x)$, the paths are at a distance
$\Theta(1)$ apart.
Using \cite[Proposition 21]{bui2025squarepackingasymptoticallysmallest}
(with minor adaptation to make it work with two sides of the boundary instead of two unit squares)
on each of the paths $\gamma_1$, $\gamma_2$, $\dots$, $\gamma_k$,
because the left and bottom edges
are sloped by $\theta$ with respect to each other, for each such path,
the total wasted area in a region near each path is $\in \Omega(\theta)$.
Since the $k$ paths are at distance $\Theta(1)$ apart, each point is only near $O(1)$ paths,
therefore the total wasted area in any packing of the parallelogram
is $\in \Omega(k \theta) = \Omega(x \theta)$, finishing the proof.
\end{proof}

\begin{figure}
    \centering
    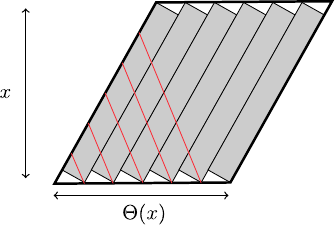
    \caption{Illustration of the tightly-packed parallelogram formed by a stack of squares packed between two parallel lines.}
    \label{fig:basic_parallelogram}
\end{figure}

Using a very similar argument, we can see that
when $m/i_m \in \Theta(1)$,
a quadrilateral with almost-vertical and almost-horizontal sides
with width $\Theta(m)$, height $\Theta(i_m)$,
and an interior angle $90^\circ - \Theta(\theta)$
must have wasted area $\Omega((m+i_m) \cdot \theta)$.
Therefore, for the quadrilaterals formed by our packing method described in this section,
our method is asymptotically optimal up to a multiplicative constant factor.

\section{Packing a Right Trapezoid}
\label{sec:right_trapezoid_packing}

We consider a right trapezoid with height $x$, base $\Theta(x^{\beta})$, slope of right edge $\Theta(x^{-\gamma})$,
where $\beta$ and $\gamma$ are some positive constants.
We will describe a packing method that results in
the wasted area being $\Theta(x^{1-\gamma/2})$ under some choices of $\beta$ and $\gamma$.

This packing method is illustrated in \Cref{fig:main_trapezoid_2}.
Informally, we first pack the cyan quadrilateral $E_0 G_0 M_0 N_0$ with the method in \Cref{sec:tightly_packed_quad},
leave a gap $N_0 M_0 G_1 E_1$ with integral height that can be almost perfectly packed
and segment $E_1 G_1$ has length slightly more than an integer,
which allows us to continue packing the cyan quadrilateral $E_1 G_1 M_1 N_1$ with the same method in \Cref{sec:tightly_packed_quad}.
In doing so, the edge $M_0 N_0$, and thus $E_1 G_1$, is slightly more tilted than $E_0 G_0$. Similarly, $E_2 G_2$ is slightly more tilted than $E_1 G_1$, etc.

Later, we will specialize to
$\gamma = \frac{1}{2}$ and $\beta = 1-\frac{\gamma}{2}=\frac{3}{4}$.
All notations used in this section are listed in \Cref{summary_notations_2}.

\begin{table}
\begin{tabular}{ccc}
    Symbol                            & Note \\ \hline
    $\beta$ & Top edge has length $\Theta(x^\beta)$ \\
    $\sigma_1$ & Slope of right edge \\
    $\gamma$ & $\sigma_1 \in \Theta(x^{-\gamma})$ \\
    $\theta$, $\theta_i$ & Slope of horizontal stacks \\
    $\omega$ & $\theta \in \Theta(x^{-\omega})$ \\
    $\sigma_2$ & $\sigma_2 \in \Theta(x^{-(4 \omega-\gamma)})$
\end{tabular}
\caption{Summary of notations in \Cref{sec:right_trapezoid_packing}.}
\label{summary_notations_2}
\end{table}

\subsection{Details of the Packing Method}

\begin{figure}
    \centering
    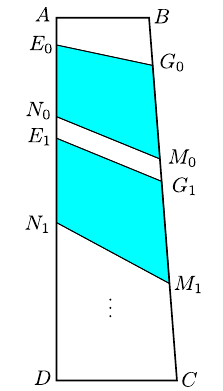
    \caption{Illustration for packing of a right trapezoid.
    Each cyan region is packed according to the method in \Cref{sec:tightly_packed_quad}.}
    \label{fig:main_trapezoid_2}
\end{figure}

We perform the following procedure. See \Cref{fig:main_trapezoid_2} for demonstration.

First, define $\omega = \frac{\gamma}{2}$ and pick $\theta \in \Theta(x^{-\omega})$, the exact constant factor to be decided later.
Then pick $E_0$ on segment $AD$, its exact position to be determined later.
Pick $G_0$ on segment $BC$ such that the angle $A E_0 G_0$ is $(90+\theta)^\circ$.
\begin{proposition}
    \label{prop:Ei_location}
    The set of locations of $E_0$ on segment $AD$ such that segment $E_0 G_0$ has length $\tan \theta$
    more than an integer is a discrete set of points, each $\Theta(x^{\gamma})$ spaced apart from its nearest neighbor.
\end{proposition}

This comes immediately from the fact that line $BC$ makes with line $AD$ an angle of $\Theta(x^{-\gamma})$.
Therefore, there exists a choice of $E_0$ such that segment $A E_0$ has length $\Theta(x^{\gamma})$.

Using this choice of $E_0$, construct a tightly-packed quadrilateral
(as described in \Cref{sec:tightly_packed_quad})
right below segment $E_0 G_0$. Let the two bottom vertices of it
be $M_0$ and $N_0$ respectively, with $M_0$ on $BC$ and $N_0$ on $AD$.

Let $\theta$, $\sigma_1$, $\sigma_2$ be as in \Cref{sec:tightly_packed_quad},
where we use quadrilateral $E_0 G_0 M_0 N_0$ here in place of
the quadrilateral $ABCD$ in \Cref{sec:tightly_packed_quad}.
Notice that the $\theta$ as in \Cref{sec:tightly_packed_quad} agrees with our definition of $\theta$ at the start of this section,
and $\sigma_1 \in \Theta(x^{-\gamma})$ is the slope of the right edge.

Then, define $\theta_1$ to be the angle that $N_0 M_0$ makes with the horizontal line. We have
$\theta_1 = \theta + \sigma_2$.

We will construct points $E_1$ and $G_1$ on segments $N_0 D$ and $BC$ respectively such that $E_1 G_1$ is parallel to $N_0 M_0$.

Arguing similar to \Cref{prop:Ei_location}, the set of possible locations of $E_1$ such that the length of $E_1 G_1$
is $\tan \theta_1$ more than an integer is a discrete set of points, each $\Theta(x^{\gamma})$ spaced apart from its nearest neighbor. Thus we can pick $E_1$ being one of those points such that segment $N_0 E_1$ has length $\Theta(x^{\gamma})$.

Move $E_1$ slightly downwards (by a length of $O(1)$) so that the segment $N_0 E_1$ has integral length. Move $G_1$ accordingly, keeping $E_1 G_1$ parallel to $N_0 M_0$.

Construct a quadrilateral $E_1 G_1 M_1 N_1$ similar to above.
Because of the movement of $E_1$ above, we might have to trim away a vertical strip of height approximately $E_1 N_1$
and width $O(x^{-\gamma})$ from the left side of the quadrilateral, near segment $E_1 N_1$.
Since the height of the right trapezoid is $x$, the total wasted area
caused by this trimming is $O(x^{1-\gamma})$. This is asymptotically smaller than the
$\Theta(x^{1-\gamma/2})$ that will be shown later, thus can be safely ignored.

Then, keep constructing quadrilaterals $\{ E_i G_i M_i N_i \}_i$ following the same procedure until the bottom edge is reached.
Let $k$ be the maximum integer for which the quadrilateral $E_k G_k M_k N_k$ was constructed
and entirely contained in the right trapezoid $ABCD$.

Next, for each gap $N_i M_i G_{i+1} E_{i+1}$,
we fill it with vertical stacks of unit squares with height equal to the length of segment $N_0 E_1$.
An illustration is in \Cref{fig:fill_gaps}.

\begin{figure}
    \centering
    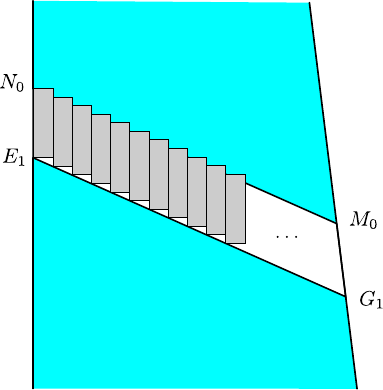
    \caption{Method of filling in the gaps $N_i M_i G_{i+1} E_{i+1}$. Illustrated with $i=0$.}
    \label{fig:fill_gaps}
\end{figure}

Note that there is some overlap with the region above $N_0 M_0$, but this is fine---%
recall the construction in \Cref{sec:tightly_packed_quad},
these columns merely extends the $T_{\bullet, j}$ columns.
The right end near $M_0 G_1$ is sloped, they can be filled in naively.

The top and bottom part can be filled in naively with waste proportional to the total perimeter
of quadrilaterals $A B G_0 E_0$ and $N_k M_k C D$.

\subsection{Analysis of the Wasted Area}

First, we consider the quadrilateral $E_0 G_0 M_0 N_0$.
We use the notation $\theta$, $\sigma_1$, $\sigma_2$, $m$ and $i_m$ for width and height as in \Cref{sec:tightly_packed_quad}.

\begin{proposition}
    $\sigma_2 \in \Theta(x^{-\gamma})$ and $i_m \in \Theta(x^\beta)$.
\end{proposition}

\begin{proof}
    We have $m \approx E_0 G_0 \in \Theta(x^\beta)$.
    Since $\theta \in \Theta(x^{-\omega})$
    and $4 \sigma_1 \sigma_2 \approx \theta^4$,
    we get $\sigma_2 \in \Theta(x^{-(4\omega-\gamma)}) = \Theta(x^{-\gamma})$.
    So $i_m \approx \sqrt{\frac{\sigma_2}{\sigma_1}}m \in \Theta(x^\beta)$.
\end{proof}

Let $\theta_i$ be the angle that $E_i G_i$ makes with the horizontal line.
Note that $\theta_1$ agrees with the definition above, and $\theta_0 = \theta$.

We want to fill the whole trapezoid $ABCD$ with such quadrilaterals,
leaving the bottom region small. The following proposition describes
a sufficient condition to do that.
\begin{proposition}
    \label{prop:sufficiency_fill_whole_trapezoid}
    If $\beta \geq 1-\gamma$ and $1-\max(\beta, \gamma) \leq \frac{\gamma}{2}$,
    and $\theta$ is chosen with sufficiently small constant factor,
    then we can fill the whole trapezoid $ABCD$ with quadrilaterals as above,
    while keeping $\theta_i \in \Theta(x^{-\omega})$ for all $i$.
\end{proposition}

Here, $\beta \geq 1-\gamma$ is just a convenient assumption, so that each length $E_i G_i$
is $\in \Theta(x^\beta)$. (Otherwise for sufficiently large $i$, length of $E_i G_i$
may grow to $\Theta(x^{1-\gamma})$. See also discussion in \Cref{subsec:triangle_like_trapezoid}.)

Note that the height of each quadrilateral plus a gap is
$\Theta(x^\beta+x^\gamma)$.
As such, assuming the height of the quadrilaterals remains roughly the same,
we need $\Theta(x^{1-\max(\beta, \gamma)})$ such quadrilaterals.

If $1-\max(\beta, \gamma) \leq \frac{\gamma}{2}$,
the angle $\theta_i$
remains in $\Theta(x^{-\omega})$ since the sum of $\sigma_2$ values over all quadrilaterals
are $\in O(x^{-\gamma}) \cdot O(x^{\gamma/2}) = O(x^{-\gamma/2}) = O(x^{-\omega})$.

The formal proof follows.
We have $\sigma_1 \in \Theta(x^{-\gamma})$,
and the length of the top edge is $\in \Theta(x^\beta)$ where $\beta = 1-\frac{\gamma}{2}$.

Because we only need to consider sufficiently large $x$,
we can assume there are constants $0<l<u$ such that
$l \cdot x^{-\gamma} < \sin \sigma_1 < \sigma_1 < u \cdot x^{-\gamma}$,
and the top edge has length $> l \cdot x^\beta+1$.

Let $d = 0.12 \cdot l^2 u^{-1}$.
Pick $\theta_0 = \arctan(d \cdot x^{-\omega})$.

Now assume $x$ is large enough such that the conditions above hold,
and in addition, $u x^{-\gamma}+2d x^{-\omega} < \frac{1}{8}$.

For each $i \geq 1$, define $\theta_i =
\arctan\Bigl(\frac{1-\cos \theta_{i-1}}{\sec(\theta_{i-1}+\sigma_1)\sin \sigma_1}
\cdot(\sec \theta_{i-1}-1)+\tan \theta_{i-1}\Bigr) $.

Then
\[
    \tan \theta_i - \tan \theta_{i-1} = \frac{1}{\sec(\theta_{i-1}+\sigma_1)}\cdot \frac{(1-\cos \theta_{i-1}) \cdot (\sec \theta_{i-1}-1)}{\sin \sigma_1}.
\]

\begin{lemma}
    If $0<\theta+\sigma_1<\frac{1}{8}$ then $1 < \sec(\theta+\sigma_1) < 1.01$.
\end{lemma}

\begin{lemma}
    If $0<\theta<\frac{1}{8}$ then $(1-\cos \theta)(\sec \theta-1) < \frac{(\tan \theta)^4}{4}$.
\end{lemma}

\begin{lemma}
    If $0<\theta<\frac{1}{8}$ then $1-\cos \theta > 0.49 (\tan \theta)^2$.
\end{lemma}

The three lemmas above can be verified with a combination of interval arithmetic and
differentiating the expression, noticing equality holds at $\theta = 0$.

\begin{lemma}
    \label{technical_lemma_6}
    For $0 \leq i \leq \lfloor \frac{l}{4d^3} x^\omega\rfloor+1$,
    $\tan \theta_i \leq \tan \theta_0 + \frac{4d^4 i}{l} x^{-\gamma}$.
\end{lemma}
\begin{proof}
    The statement is true for $i = 0$.

    We prove by induction.
    Assume $\tan \theta_{i-1} \leq \tan \theta_0 + \frac{4 d^4 (i-1)}{l} x^{-\gamma}$.
    Then since $i-1 \leq \lfloor \frac{l}{4 d^3} x^\omega \rfloor$,
    $\tan \theta_{i-1} \leq 2d \cdot x^{-\omega} < \frac{1}{8}$,
    so $\tan \theta_i - \tan \theta_{i-1} < \frac{(\tan \theta_{i-1})^4}{4 \sin \sigma_1} < \frac{(2d x^{-\omega})^4}{4 l x^{-\gamma}}
    \leq \frac{4 d^4}{l} x^{-\gamma}$, so we're done.
\end{proof}

\begin{lemma}
    The total height of the first $\bigl\lceil\frac{l}{4d^3} x^\omega\bigr\rceil$ trapezoids
    is $\geq x$.
\end{lemma}
\begin{proof}
    Consider a particular trapezoid with slope of top edge being $\theta_i$, where $0 \leq i < \bigl\lceil \frac{l}{4d^3}x^\omega \bigr\rceil$.
    The number of columns $m$ is $\geq$ the length of the top edge, which is $>l \cdot x^\beta+1$.
    Therefore the number of rows is $\geq \Bigl\lceil \frac{(m-1)(1-\cos \theta_i)}{\sec(\theta_i+\sigma_1)\sin \sigma_1}\Bigr\rceil+1$.
    This is $\geq \frac{l \cdot x^\beta \cdot(1-\cos \theta_i)}{\sec(\theta_i+\sigma_1) \sin \sigma_1}
    \geq \frac{l \cdot x^\beta \cdot 0.49 (\tan \theta_i)^2}{1.01 \cdot u \cdot x^{-\gamma}}$.
    Since each $\tan \theta_i$ is $\geq \tan \theta_0 = d \cdot x^{-\omega}$,
    each number of rows above is $\geq 0.48 x^{1-\omega} d^2 \frac{l}{u}$.

    Therefore the total height is $\geq \frac{l}{4d^3}x^\omega \cdot 0.48 x^{1-\omega} d^2 \frac{l}{u} = 0.12 \frac{l^2}{du} x = x$.
\end{proof}

We have shown that we need no more than $\bigl\lceil\frac{l}{4d^3} x^\omega\bigr\rceil$ quadrilaterals to
fill the whole trapezoid $ABCD$, and the first
$\bigl\lceil\frac{l}{4d^3} x^\omega\bigr\rceil$ quadrilaterals have $\theta_i \in \Omega(x^{-\omega})$,
so \Cref{prop:sufficiency_fill_whole_trapezoid} is proven.

Next, we analyze the wasted area.

The top and bottom part have wasted area $O(x^\beta+x^\gamma)$.

There are $O(x^{1-\max(\beta, \gamma)})$ such quadrilaterals, and roughly as many gaps.
For each quadrilateral, the waste is $O((m+i_m) \theta) = O(x^{\beta-\omega})$.
For each gap, the waste
caused by the naive filling near the segments $M_i G_{i+1}$
is $O(x^\gamma)$, and the waste below each vertical stack can be discounted because they're equal to the
amount of space reused by overlapping with the quadrilateral above $N_0 M_0$.

Adding them up, the total wasted area is
\[
    O\bigl( x^\beta + x^\gamma + x^{1-\max(\beta, \gamma)} \cdot (x^{\beta-\gamma/2} + x^\gamma) \bigr).
\]

Finally, specialize to $\gamma = \frac{1}{2}$ and $\beta = \frac{3}{4}$.
All hypotheses are satisfied, and the total wasted area is $O(x^{3/4})$.

\section{Reduction from Packing a Square to Packing a Right Trapezoid}
\label{trapezoid_reduction}

We will state a proposition which describes a packing method that allows one to reduces the problem of packing 
a square to the problem of packing a right trapezoid, special cases of which has already been used several times in previous works.

Define $W_{\beta, \epsilon}(x) = x^{\frac{2 \beta}{2 \beta+1}} \log^{\frac{\epsilon}{2 \beta+1}} x$.
We would like to note that $\frac{2 \beta}{2 \beta+1}$ is the harmonic mean of $\frac{1}{2}$ and $\beta$.
\begin{proposition}
    \label{prop:type2_to_type1}
    If there exists
    real $\frac{1}{2}<\beta<1$,
    real $0<\nu<\beta+\frac{1}{2}$,
    real $\epsilon$
    such that
    for all real $m$,
    for all $w \in \Theta(m^\nu)$,
    the right trapezoid with height $m$, smaller base $w$, larger base $w+\Theta(\sqrt m)$
    can be packed with wasted area $O(m^\beta \log^\epsilon m)$,
    then $W(x) \in O(W_{\beta, \epsilon}(x))$.
\end{proposition}

This trapezoid is roughly the same as a ``type 2'' shape in \cite{Chung_2009, wang2016newresultpackingunit},
except that we make the constant factor implicit rather than explicit.

The wasted area is $O(m^\beta \log^\epsilon m + \frac{x}{\sqrt m})$.
By selecting $m = (x \log^{-\epsilon}x)^{\frac{2}{2\beta+1}}$,
the quantity above is minimized, being $O(x^{\frac{2 \beta}{2 \beta+1}} \log^{\frac{\epsilon}{2 \beta+1}} x) = W_{\beta, \epsilon}(x)$.

Note that $\nu$ cannot be too large otherwise it may happen that $m^\nu > x$.

We see this being applied in previous results as follows.
Note that $\frac{4-\sqrt 2}{7}<1$.

\begin{center}
\begin{tabular}{cccc}
    Article                            & $m^\beta \log^\epsilon m$      & Choice of $m$& $W_{\beta, \epsilon}(x)$ \\ \hline
    \cite{Erdos_1975}                  & $m^{7/8}$                       & $x^{8/11}$   & $x^{7/11}$ \\
    \cite{Chung_2009}
                                       & $m^{\frac{2+\sqrt 2}{4}} \log m$& $x^{2-2 \alpha}$, with $\alpha=\frac{3+\sqrt 2}{7}$ & $x^{\frac{3+\sqrt 2}{7}}\log^{\frac{4-\sqrt 2}{7}} x$
    \\
    \cite{wang2016newresultpackingunit}& $m^{5/6}$                       & $x^{3/4}$         & $x^{5/8}$
\end{tabular}
\end{center}

In this article, when $\beta = \frac{3}{4}$, we get:
\begin{theorem}
    $W(x) \in O(x^{3/5})$.
\end{theorem}

That proves the claim in the introduction.

\section{Discussion}

\subsection{Symmetry of the Construction}

We discuss a symmetry in the construction described in \Cref{sec:tightly_packed_quad}.

Note that if we shift each $T_{\bullet, j}$ stack down until they touch segment $CD$,
the small rectangles in $W_6$ will disappear, instead, some small rectangles
will appear above the stacks $T_{\bullet, j}$---%
and we see the symmetry of the construction between
$W_1 \leftrightarrow W_3$,
$W_2 \leftrightarrow W_6$,
$W_4 \leftrightarrow W_5$,
$\sigma_1 \leftrightarrow \sigma_2$,
$S \leftrightarrow T$.

The symmetry is also shown in the following formula that relates $\theta$, $\sigma_1$ and $\sigma_2$:
\[
\sec(\theta+\sigma_1) \sin \sigma_1 \cdot
\sec(\theta+\sigma_2) \sin \sigma_2
= (1-\cos \theta)^2.
\]

And the following formula (this can be derived from $i_m \approx \frac{\theta^2}{2 \sigma_1}m$):
\[ \frac{i_m}{\sqrt{\sigma_2}} \approx \frac{m}{\sqrt{\sigma_1}}. \]

While the square root may look weird, a better way to look at it is the following:
if $\theta$ and $m$ are fixed, in order to double $\sigma_2$, you need to multiply $i_m$ by roughly the same factor $2$. In formula: $i_m \approx \frac{2\sigma_2}{\theta^2} m$.

Note that the area of $W_6$ does not have a $O(\sigma_2 m)$ term likely because our analysis is not exactly symmetric, in particular the width of $W_2$ is measured being perpendicular to $AD$ but the height of $W_6$ is measured being parallel to $AD$ instead of perpendicular to $AB$. Nevertheless, the result is not affected.

\subsection{Dual of the Packing Method}
\label{subsec:dual_packing}

We would like to note that there is a natural dual
to the method in \Cref{sec:tightly_packed_quad}.
See \Cref{fig:tight_packing_dual}.

\begin{figure}
    \centering
    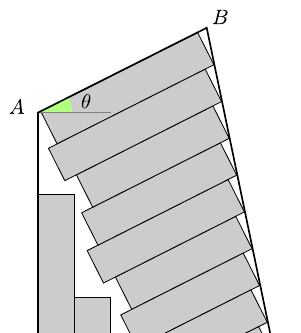
    \caption{Dual of the packing described above.}
    \label{fig:tight_packing_dual}
\end{figure}

Here, the horizontal stacks are angled upwards instead of downwards, and it is necessary that $\theta > \sigma_1$.
We don't analyze this configuration because it appears to be not particularly useful.

\subsection{Limitation of the Packing Method}
\label{subsec:limitation}

Consider the packing method in \Cref{sec:right_trapezoid_packing}.
Suppose we choose some value of $\omega$ different from $\frac{\gamma}{2}$.

If $\omega < \frac{\gamma}{2}$, then $\theta > x^{-\gamma/2}$ for large enough $x$,
then the waste comes from the quadrilaterals alone is $\geq \Omega(x \cdot \theta)$
which is more than $x^{1-\gamma/2}$.

If $\omega > \frac{\gamma}{2}$,
the height of each quadrilateral is $\Theta(x^{\beta+\gamma-2 \omega})$
which is much \emph{less} than the width $\Theta(x^\beta)$.
Consequently, the total perimeter of the quadrilaterals is much more than the sum of the heights.
Assume at least a constant factor of the height comes from the quadrilaterals instead of the gaps
(that is $\beta+\gamma-2 \omega \geq \gamma$), then the wasted area (again comes from the quadrilaterals alone)
is
$\Theta(x^{2 \omega-\gamma} \cdot x \cdot \theta) = \Theta(x^{1+\omega-\gamma})$,
and $1+\omega-\gamma > 1-\frac{\gamma}{2}$.

As such, it appears that $\Omega(x^{1-\gamma/2})$ is a natural lower bound for our method.
And, with $\beta = \frac{1}{2}$, our method of packing naively the top/bottom/right side area in \Cref{sec:right_trapezoid_packing}
does not hurt us, since naive packing already gives $O(x^{1-\gamma/2})$.

There is another interpretation for $\Omega(x^{1-\gamma/2})$.
Consider \Cref{fig:main_trapezoid_subdivision_lower_bound}.
If the trapezoid $ABCD$ is subdivided into smaller right trapezoids by cutting horizontally,
such that each triangle formed by drawing a vertical line from the top right corner
to the bottom edge has area $\Theta(1)$
(colored \textcolor{cyan}{cyan} in the figure),
then you need $\Theta(x^{1-\theta/2})$ such triangles.
It appears unlikely that it is possible to pack each of these small trapezoid with average wasted area $o(1)$ each.

\begin{figure}
    \centering
    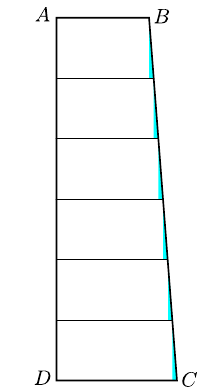
    \caption{Horizontally cut a right trapezoid into $\Theta(x^{1-\gamma/2})$ smaller right trapezoids.}
    \label{fig:main_trapezoid_subdivision_lower_bound}
\end{figure}

As such, we make the following conjecture:
\begin{conjecture}
    \label{conjecture_lower_bound_trapezoid}
    For $0<\gamma<1$,
    a right trapezoid with height $x$ and difference between two bases $\Theta(x^{1-\gamma})$
    cannot be packed in $o(x^{1-\gamma/2})$.
\end{conjecture}

Unfortunately, the most straightforward method to prove this---proving that each such right trapezoid
has wasted area $\Omega(1)$---doesn't work (when $\beta>\gamma>0$),
as depicted in the configuration in \Cref{fig:counterexample_naive_attempt}.

We consider the right trapezoid colored gray.
Draw a blue diagonal line as in the figure,
then draw a yellow line parallel to the bottom side
and a red line parallel to the top side.
Pack each region with stacks of squares with tilt $\Theta(x^{-\beta/2})$.

Then the difference in the $x$-coordinate of the two endpoints of the blue diagonal line
is $\Theta(x^{\beta/2})$, which is larger than the distance $\Theta(x^{\gamma/2})$ marked on the figure.
The total wasted area inside the gray right trapezoid is then
$O(x^{\gamma/2}\cdot (x^{-\beta/2} + x^{-\gamma})) \subseteq o(1)$.

\begin{figure}
    \centering
    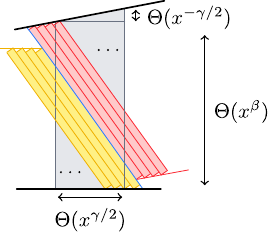
    \caption{Example of packing where the naive attempt to prove Conjecture~\ref{conjecture_lower_bound_trapezoid} fails.}
    \label{fig:counterexample_naive_attempt}
\end{figure}

Nonetheless, we hope it is possible to adapt the methods in \cite{roth1978inefficiency}.

However, getting from the $x^{1-\gamma/2}$ barrier to a better lower bound
for square packing is still highly nontrivial,
since there is no reason why a packing method must use a reduction of the form \Cref{prop:type2_to_type1}.
See \Cref{fig:invariant_insufficiency} for an illustration.

\begin{figure}
    \centering
    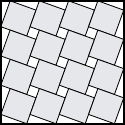
    \caption{An illustration of the waste area being $O(x)$ without any of the known bottlenecks.}
    \label{fig:invariant_insufficiency}
\end{figure}

\subsection{Packing Other Shapes}

We have shown in \Cref{sec:right_trapezoid_packing} that for certain values of $\beta$ and $\gamma$,
a right trapezoid with height $x$, width $\Theta(x^\beta)$, slope of right angle $\Theta(x^{-\gamma})$
can be packed with wasted area $\Theta(x^{1-\gamma})$.
(Note that, assuming $\beta \geq 1-\gamma$, when $\beta$ is too small, the wasted area is dominated by
$x^{1-\beta/2}$, and when $\beta$ is too large the wasted area is dominated by $x^{\beta-1/2}$.)

Our method as is only works for $\gamma$ up to $\frac{1}{2}$.
Additional considerations, such as packing the top/bottom/right area of the gap intelligently
by recursively using the original method
(the method of recursive packing can be found in \cite{Chung_2009})
may make it work for higher $\gamma$.
Further research is needed to determine the behavior at various values of $\gamma$.

In fact, we make the following conjecture:

\begin{conjecture}
    For any value $\beta \geq 0$ and $\gamma \geq 0$, the optimal exponent in the wasted area asymptotic
    of packing a trapezoid with height $x$, width $\Theta(x^\beta)$, slope of right angle $\Theta(x^{-\gamma})$
    is
\[
    \max\Bigl(1-\frac{\max(\beta,1-\gamma)}{2}, 1-\frac{\gamma}{2}, \beta-\frac{1}{2}, \frac{3}{5}\Bigr).
\]
\end{conjecture}

The wasted area of the trivial packing method is $\Theta(x (x-\lfloor x \rfloor))$,
therefore, when $x-\lfloor x \rfloor \in o(x^{-2/5})$
then it is possible to get the wasted area $o(x^{3/5})$.
That poses the question:
\begin{question}
    Is it possible to get wasted area $o(x^{3/5})$
    when $x-\lfloor x \rfloor \in \Theta(x^{-2/5})$?
\end{question}
When $x-\lfloor x \rfloor \in \Theta(x^{-2/5})$,
\cite{roth1978inefficiency} gives the lower bound $x^{3/10}$.

Our reason for focusing on right trapezoid is that the problem of packing arbitrary almost-rectangular quadrilaterals
can often be reduced to packing right trapezoids
by dividing such quadrilaterals into a small number of right trapezoids.
The reduction may not be optimal, however.

\subsection{Inspiration for the Primitive Tightly-packed Quadrilateral}
\label{subsec:tightly_packed_quad_inspiration}

Here we explain the connection between \cite{wang2016newresultpackingunit} and our construction in \Cref{sec:tightly_packed_quad}.

In \cite{wang2016newresultpackingunit}, noticing that the slope of each horizontal stack is roughly $\sqrt{\frac{2 \delta}{w}}$,
when $\delta$ changes by a small amount, say $w^{-0.9}$,
in order to make $\frac{2 \delta}{w}$ remains roughly the same,
the denominator should be scaled by roughly the same factor as the numerator.

Suppose $\delta \in \Theta(1)$.
Then $\delta - \Theta(w^{0.1}) \approx \delta \cdot (1-\Theta(w^{-0.9}))$.
We want to scale the numerator by the same factor, namely changing the numerator
from $w$ to $w \cdot (1-\Theta(w^{-0.9}))$, which is $w - \Theta(w^{0.1})$.

In this article, we made two modifications:
\begin{itemize}
    \item we change $w$ (and thus the slope) gradually, instead of in bulk;
    \item we work backward:
        instead of
        using $\delta$ and the desired slope to determine $w$ as in \cite{wang2016newresultpackingunit},
        we use the slope of the stack of squares and $\delta$ to determine $w$.
        As such, we can ensure the slope difference is exactly zero, at the cost of a small wasted area elsewhere.
\end{itemize}

\subsection{On Packing of Triangle-like Right Trapezoid}
\label{subsec:triangle_like_trapezoid}

Consider a right trapezoid $T(x, \Theta(x^\beta), \Theta(x^{-\gamma}))$ where $\beta>0$ and $\gamma>0$.
Even though $\gamma>0$, it is not necessarily true that the trapezoid will look approximately like
a rectangle.
This is because the top (smaller) side has length $\Theta(x^\beta)$, the bottom (larger) side
has length $\Theta(x^\beta)+\Theta(x^{1-\gamma})$, if $1-\gamma>\beta$,
then the right trapezoid in fact looks like a triangle.
See \Cref{fig:triangle_like_trapezoid}.

\begin{figure}
    \centering
    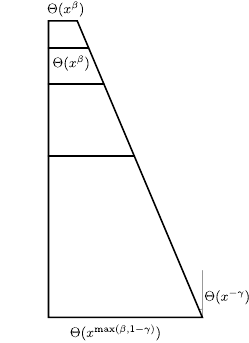
    \caption{Illustration of a triangle-like right trapezoid.}
    \label{fig:triangle_like_trapezoid}
\end{figure}

We would like to note that it is possible to split a triangle-like trapezoid to $O(\log x)$ right trapezoids,
with each of them having bottom side no more than twice the top side.
This would allow us to only focus on packing trapezoids with ratio of two bases $\in \Theta(1)$.

In fact, we conjecture the following:

\begin{conjecture}
    The strategy of packing a triangle-like right trapezoid by first subdividing into $O(\log x)$
    right trapezoids as above, then packing each of them optimally,
    is asymptotically no worse than the optimal strategy.
\end{conjecture}

\section{Conclusion}

We have shown that the wasted area when packing
a large square with side length $x$ can be as small as $O(x^{3/5})$.
Further research is needed to prove or disprove various bounds,
such as the bound $x^{1-\gamma/2}$ pointed out earlier,
and extending the result to non-square shapes.

\section{Acknowledgements}

We would like to thank anonymous reviewers for their enthusiasm, and for many
helpful suggestions to improve the manuscript.

\bibliography{sn-bibliography}

\clearpage

\appendix
\crefalias{section}{appendix}

\section{Alternative Method of Wasted Area Calculation}
\label{alternative_wasted_area_calculation}

In this section, we try to show that the area of the trapezoid $ABCD$
in \Cref{remark:alternative-wasted-area-calculation}
is $\leq m \cdot i_m + O(\theta \cdot (m + i_m) + 1)$.

Draw segments $BH$, $CJ$, $CK$ perpendicular to $AD$,
$CJ$ parallel to $AB$,
with $H$ and $K$ on line $AD$, $J$ on line $HB$, $L$ on line $AB$.
See \Cref{fig:calculate_segment_lengths} for an illustration.

\begin{figure}
    \centering
    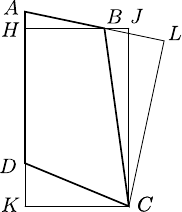
    \caption{Illustration for alternative method of calculating wasted area in \Cref{remark:alternative-wasted-area-calculation}.}
    \label{fig:calculate_segment_lengths}
\end{figure}

Intuitively, segment $AD$ has length $i_m + O(\theta)$,
the perpendicular $CK$ has length $m + O(\theta+\sigma_1)$,
therefore triangle $ACD$ has area $\frac{1}{2} (i_m + O(\theta)) \cdot (m + O(\theta+\sigma_1))$.
Similarly, segment $AB$ has length $m + O(\theta)$,
the perpendicular $CL$ has length $i_m + O(\theta+\sigma_2)$,
therefore triangle $ABC$ has area $\frac{1}{2} (i_m + O(\theta+\sigma_2)) \cdot (m+O(\theta))$.
So the total area is $m \cdot i_m + O(\theta \cdot (m+i_m)+1)$.
The difference gives the expected result.

There is an extra term $m \sigma_2$ or $i_m \sigma_1$,
but $i_m \sigma_1 \approx \frac{\theta^2}{2 \sigma_1} m \sigma_1 = \frac{\theta}{2} \cdot \theta m \leq \theta m$,
and similar for the other direction.

Let us compute the total area more formally.

We have mentioned above that segment $AB$ has length $m+\tan \theta$.
We see that the top edge of $T_{i_2, 1}$ is $\tan \theta + (i_2-1) \sec \theta$ below $A$,
so the bottom edge of $T_{i_m, 1}$ is $\tan \theta + (i_2-1) \sec \theta + (i_m-i_2+1)$
below $A$,
so segment $AD$ has length
$\tan \theta + (i_2-1) \sec \theta + (i_m-i_2+1) + (\sec \theta-1)
= i_m + \tan \theta + (\sec \theta-1) i_2
$.


Then, segment $BH$ has length $(m+\tan \theta)\cdot\cos \theta$ and
segment $AH$ has length $(m+\tan \theta)\cdot\cos \theta$.

Let $x$ be the length of segment $KC$.
This is also equal to length $HJ$,
so $BJ = x - BH$,
so $HK = JC = (x - BH) \cot \sigma_1$,
so $KD = ((x - BH) \cot \sigma_1 - DH)$.

Therefore $x = ((x - BH) \cot \sigma_1 - DH) \cot (\theta+\sigma_2)$.
Solving for $x$ gives 
\begin{align*}
    x &= \frac{BH + DH \tan \sigma_1}{1 - \tan (\theta+\sigma_2) \tan \sigma_1} \\
      &= \frac{(m+\tan \theta) \cos \theta (1-\tan \sigma_1) + \tan \sigma_1 (i_m + \tan \theta + (\sec \theta-1) i_2))}{1-\tan(\theta+\sigma_2) \tan \sigma_1} .
\end{align*}

This appears to be difficult to analyze, so let us try to analyze it in a different way.
Let $E$ be the top right corner of $\mathcal C$, $G$ be the top right corner of $\mathcal D$,
$F$ on segment $CD$ such that $EF$ is vertical, and drop perpendicular $CM$ to line $EF$.
See \Cref{fig:alt_analysis} for an illustration.
Here $\mathcal C$ and $\mathcal D$ are unit squares
defined the same way as in \Cref{fig:zoom_in_D}.

\begin{figure}
    \centering
    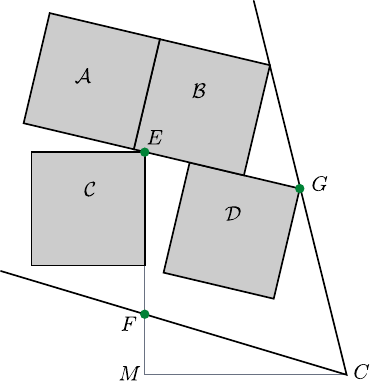
    \caption{More detailed illustration for the argument in \Cref{remark:alternative-wasted-area-calculation}.}
    \label{fig:alt_analysis}
\end{figure}

Then, segment $EG$ has length $\leq 1+\tan(\theta+\sigma_1)$ by an analysis similar to \Cref{fig:2}.
To compute the length of segment $EF$, note that:
\begin{itemize}
    \item point $D$ is $\sec \theta-1$ below the bottom left corner of $T_{i_m, 1}$,
    \item point $F$ is $(m-1)\tan(\theta+\sigma_2)$ below $D$,
    \item the square $\mathcal C$ is $\Gamma_3+\dots+\Gamma_m$ below the square $T_{i_m, 1}$,
\end{itemize}
therefore
\begin{align*}
    EF
&= 1+(m-1) \tan(\theta+\sigma_2) + (\sec \theta-1) - (\Gamma_3+\dots+\Gamma_m) \\
&= 1+(m-2) (\tan(\theta+\sigma_2)-\tan \theta) + \tan(\theta+\sigma_2) - (i_m-i_2-1) (\sec \theta-1) \\
&= 1+\Bigl( \frac{1-\Delta_1}{\Delta_2} (m-2) - (i_m-i_2-1)\Bigr) (\sec \theta-1) + \tan(\theta+\sigma_2) \\
&= 1+\biggl( \Bigl(\frac{1-\Delta_1}{\Delta_2} (m-1) - i_m\Bigr) - (\frac{1-\Delta_1}{\Delta_2} - i_2) + 1 \biggr) (\sec \theta-1) + \tan(\theta+\sigma_2) \\
&\leq 1+ 2 (\sec \theta-1) + \tan(\theta+\sigma_2) .
\end{align*}
Asymptotically, we can assume $EF<\frac{3}{2}$, $EG<\frac{3}{2}$
and both $\sigma_1$ and $\theta+\sigma_2$ are sufficiently small.
Then $GC<2$, so $CM = EG \cos \theta + GC \sin \sigma_1 \in 1+O(\theta+\sigma_1)$.

\end{document}